\documentclass[11pt]{article}
\usepackage{times}
\usepackage{latexsym,  verbatim}
\usepackage{amssymb}
\usepackage[dvips]{graphicx}
\usepackage{amstext}
\usepackage{amsmath}
\usepackage{xspace}
\usepackage{theorem}
\usepackage{fullpage}

\def \eps {\varepsilon}

\newcommand{\ignore}[1]{}
\def \poly { \text{\rm poly} }

\newtheorem{theorem}{Theorem}

\newenvironment{proof}{{\bf Proof.  }}{\hfill$\Box$}

\newtheorem{lemma}{Lemma}

\newtheorem{hyp}{Hypothesis}

\newenvironment{reminder}[1]{\smallskip
\noindent {\bf Reminder of #1  }\em}{}

\usepackage{algorithmic}
\usepackage[ruled,vlined,commentsnumbered,titlenotnumbered]{algorithm2e}
\newcommand{\Ot}{\widetilde{O}}

\def \eps {\varepsilon}

\def \poly { \text{\rm poly} }

\newcommand{\etal}{{\em et al.\ }}
\newcommand{\IN}{\textrm{\footnotesize{in}}}
\newcommand{\OUT}{\textrm{\footnotesize{out}}}

\title{Approximating the diameter of a graph}

\author{Liam Roditty\thanks{Bar Ilan University, liamr@cs.biu.ac.il. Work supported by the Israel Science Foundation
(grant no. 822/10)}~ and~ Virginia Vassilevska Williams\thanks{UC Berkeley and Stanford University, virgi@eecs.berkeley.edu. Partially supported by NSF Grants CCF-0830797 and CCF-1118083 at UC Berkeley, and by NSF Grants IIS-0963478 and IIS-0904325, and an AFOSR MURI Grant, at Stanford University.}}
\date{}

\begin{document}
\maketitle

%\begin{verbatim}
%1. Go over intro add what you think is missing from you initial intro
%2. Go over the better approximation section see it is fine
%3. Complete the better approximation section for sparse graphs
%4. add to the faster algorithm an explanation why it works in weighted graphs
%5. add to the intro something about the act that for undirected unweighted we can
%get 4/5D
%6. Connect between the theorem we cite in intro to the actual section so
%we will not have the same theorem twice.
%7. Add concluding remarks and open problems
%8. Abstract
%\end{verbatim}

\begin{abstract}
In this paper we consider the fundamental problem of approximating the diameter $D$ of directed or undirected graphs. In a seminal paper, Aingworth, Chekuri, Indyk and Motwani [SIAM J. Comput. 1999] presented an algorithm that computes in $\Ot(m\sqrt n + n^2)$ time an estimate $\hat{D}$ for the diameter of an $n$-node, $m$-edge graph, such that $\lfloor 2/3 D \rfloor \leq \hat{D} \leq D$. In this paper we present an algorithm that produces the same estimate in $\Ot(m\sqrt n)$ expected running time. We then provide strong evidence that a better approximation may be hard to obtain if we insist on an $O(m^{2-\eps})$ running time. In particular, we show that if there is some constant $\eps>0$ so that there is an algorithm for undirected unweighted graphs that runs in $O(m^{2-\eps})$ time and produces an approximation $\hat{D}$ such that $ (2/3+\eps) D  \leq \hat{D} \leq D$, then SAT for CNF formulas on $n$ variables can be solved in $O^{*}( (2-\delta)^{n})$ time for some constant $\delta>0$, and the strong exponential time hypothesis of [Impagliazzo, Paturi, Zane JCSS'01] is false.

Motivated by this somewhat negative result, we study whether it is possible to obtain a better approximation for specific cases. For unweighted directed or undirected graphs, we show that if $D=3h+z$, where $h\geq 0$ and $z\in \{0,1,2\}$, then it is possible to report in $\tilde{O}(\min\{m^{2/3} n^{4/3},m^{2-1/(2h+3)}\})$ time an estimate $\hat{D}$ such that $2h+z \leq \hat{D}\leq D$, thus giving a better than $3/2$ approximation whenever $z\neq 0$.
This is significant for constant values of $D$ which is exactly when the diameter approximation problem is hardest to solve.
For the case of unweighted undirected graphs we present an $\tilde{O}(m^{2/3} n^{4/3})$ time algorithm that reports an estimate $\hat{D}$ such that $\lfloor 4D/5\rfloor \leq \hat{D}\leq D$.
\end{abstract}

\thispagestyle{empty}
%\setpagenumber{0}

\newpage
\setcounter{page}{1}

\section{Introduction}

The diameter of a graph is the longest of all distances between vertices in the graph.
The diameter is a natural and fundamental graph parameter, and computing it efficiently has many applications (e.g.~\cite{CFM10}).
%The search for fast diameter algorithms has a long history.
Essentially, the only known way to determine the diameter of a graph with arbitrary edge weights is to compute the distances between all pairs of vertices in the graph, that is, to solve the all-pairs shortest paths problem (APSP), and then to find the maximum distance.
Because of this, some researchers have conjectured that APSP and diameter in weighted graphs may be equivalent in some sense (e.g.~\cite{focsy} and~\cite{chung}).
The fastest algorithms for computing APSP and hence for computing the diameter for directed or undirected graphs on $n$ nodes and $m$ edges with arbitrary edge weights and no negative cycles have a running time of
%This approach gives rise to
$O(\min\{n^3\log\log^3 n / \log^2 n, mn+n^2\log\log n\})$~\cite{chan07,Pettie04}.

%For the special case of dense unweighted graphs, there are faster algorithms for diameter based on fast matrix multiplication.
For the special case of dense directed or undirected unweighted graphs, one can compute the diameter by reducing its computation to fast matrix multiplication, thus obtaining $\tilde{O}(n^\omega)$ time algorithms, where $\omega<2.38$ is the matrix multiplication exponent~\cite{cw90,stothers,v12}.
In fact, any known algorithm for diameter in dense $n$-node unweighted graphs running in $T(n)$ time can also be used to compute the Boolean product of two $n\times n$ Boolean matrices in $O(T(n))$ time. This lead to conjectures~\cite{chung,aingworth} that computing the diameter in dense unweighted graphs and Boolean matrix multiplication (BMM) may be equivalent.

%For instance, APSP can be solved in $\tilde{O}(n^\omega)$ time in undirected unweighted graphs~\cite{seidel} where $\omega<2.38$ is the matrix multiplication exponent, and in $O(n^{2.575})$ time in directed unweighted graphs~\cite{zwick02}.
For the special case of sparse directed or undirected unweighted graphs, the best known algorithm for both APSP and diameter does breadth-first search (BFS) from every node and hence runs in $O(mn)$ time. For sparse graphs with $m=O(n)$, the running time is $\Theta(n^2)$ which is natural for APSP since the algorithm needs to output $n^2$ distances. However, for the diameter the output is a single integer, so it is not immediately clear why one should spend $\Omega(n^2)$ time to compute it. In this paper, we show somewhat surprisingly, that breaking this seeming $n^2$ barrier would have major consequences for the complexity of NP-hard problems such as SAT.

A natural question is whether one can get substantially faster algorithms for the diameter by settling for an approximation.
A $c$-approximation algorithm for the diameter $D$ of a graph for $c\geq 1$ provides an estimate $\hat{D}$ such that $D/c\leq \hat{D}\leq D$.
It is well known that a $2$-approximation for the diameter in directed or undirected graphs with nonnegative weights is easy to achieve in $\tilde{O}(m)$ time using Dijkstra's algorithm from and to an arbitrary node.
Dor, Halperin and Zwick~\cite{DorHZ00} showed that any $(2-\eps)$-approximation algorithm for APSP even in unweighted graphs running in $T(n)$ time would imply an $O(T(n))$ time for BMM, and hence apriori it could be that $(2-\eps)$-approximating the diameter of a graph may also require solving BMM.

In their seminal paper, Aingworth, Chekuri, Indyk and Motwani~\cite{aingworth}
showed that it is in fact possible to get a subcubic $(2-\eps)$-approximation algorithm for the diameter of graphs with nonnegative weights
without resorting to fast matrix multiplication.
 %break this cubic running time barrier, without using Fast Matrix Multiplications (FMM) algorithms, by settling for an approximation of the diameter.
In particular, they designed an
%showed that it is possible to compute in
$\tilde{O}(m\sqrt{n}+n^2)$ time algorithm computing an estimate $\hat{D}$ that satisfies $\lfloor 2D/3 \rfloor \leq \hat{D}\leq D$. Their algorithm has several important and interesting properties. It is the only known algorithm for approximating the diameter
 polynomially faster than $O(mn)$ for every $m$ that is superlinear in $n$. It always runs in truly subcubic time even in dense graphs, and
%by factor better than $2$ that does not use fast matrix multiplication and
does not explicitly compute
all-pairs approximate shortest paths.
A natural question
%is whether there exist fast constant approximation algorithms for the diameter, and more specifically,
is whether there is an almost linear time approximation scheme for the diameter problem: an algorithm that for any constant $\eps>0$ runs in  $\tilde{O}(m)$ time and returns an estimate $\hat{D}$ such that $(1-\eps)D \leq \hat{D}\leq D$. Such an algorithm would be of immense interest, and has not so far been explicitly ruled out, even conditionally.
%Any $(3/2-\eps)$-approximation algorithm for the diameter would be able to
%%Such an algorithm allows to
%distinguish between graphs of diameter $2$ and $3$. Hence, understanding the complexity of this simple $2$ vs $3$ problem is crucial for obtaining better diameter approximation algorithms.
In this paper we give strong evidence that a fast $(3/2-\eps)$-approximation algorithm for the diameter may be very hard to find, even for undirected unweighted graphs. We show:
%More precisely, we show:
\begin{theorem}\label{thm:hard}
Suppose there is a constant $\eps>0$ so that, there is a $(3/2-\eps)$-approximation algorithm for the diameter in $m$-edge undirected unweighted graphs that runs in $O(m^{2-\eps})$ time for every $m$. Then, SAT for CNF formulas on $n$ variables can be solved in $O^{*}( (2-\delta)^{n})$ time for some constant $\delta>0$.
%that can distinguish between $m$-edge graphs of diameter $2$ and $m$-edge graphs of diameter at least $3$ in $O(m^{2-\eps})$ time for some $\eps>0$ and for every $m$, can be used to solve the CNF-SAT problem faster
\end{theorem}

The fastest known algorithm for CNF-SAT is the exhaustive search algorithm that runs in $O^{*}(2^n)$ time by trying all possible $2^n$ assignments to the variables. It is a major open problem whether there is a faster algorithm. Several other NP-hard problems are known to be equivalent to CNF-SAT so that if one of these problems has a faster algorithm than exhaustive search, then all of them do~\cite{cygan}. Hence, our result also implies that if the diameter can be approximated fast enough, then also problems such as Hitting Set,
Set Splitting, or NAE-SAT, all seemingly unrelated to the diameter, can be solved faster than exhaustive search.
The strong exponential time hypothesis (SETH) of Impagliazzo, Paturi, and Zane~\cite{ipz1,ipz2} implies that there is no improved $O^{*}((2-\delta)^n)$ time algorithm for CNF-SAT, and hence our result also implies that there is no $(3/2-\eps)$-approximation algorithm for the diameter running in
$O(m^{2-\eps})$ time unless SETH fails. (We elaborate on this hypothesis later on in the paper.)

We prove Theorem~\ref{thm:hard} by showing that an $O(n^{2-\eps})$ time, $(3/2-\eps)$-approximation algorithm for the diameter in sparse graphs with $m=O(n)$ would imply an $O^{*}( (2-\delta)^{n})$ time CNF-SAT algorithm.
%It implies that unless SETH fails, $O(n^2)$ time is essentially required for $(2/3+\eps)$-approximating the diameter in sparse graphs, within $n^{o(1)}$ factors.
This implies that unless SETH fails, $O(n^2)$ time is essentially required to get a $(3/2-\eps)$-approximation algorithm for the diameter in sparse graphs, within $n^{o(1)}$ factors. Hence, within $n^{o(1)}$ factors, the time for $(3/2-\eps)$-approximating the diameter in a sparse graph is the same as the time required for computing APSP exactly!

Even more concretely, we prove Theorem~\ref{thm:hard} by showing that distinguishing whether the diameter of a given undirected unweighted graph is $2$ or at least $3$ fast enough
%in $O(n^{2-\eps})$ time for some $\eps>0$
would imply an improved SAT algorithm. (Any $(3/2-\eps)$-approximation algorithm for the diameter would be able to
distinguish between graphs of diameter $2$ and $3$.)
%The simplest case for which no nontrivial algorithms are known is the case of distinguishing whether the diameter of a given undirected unweighted graph is $2$ or at least $3$.
The fastest algorithms for this special case of the diameter problem still run in $\tilde{O}(\min\{mn, n^\omega\})$ time, and several papers have asked whether one can do better~\cite{chung,aingworth}. In 1987, Chung~\cite{chung} actually conjectured that this problem may be equivalent to BMM, so that any subcubic algorithm for it can be converted to a subcubic algorithm for BMM. Aingworth \etal~\cite{aingworth} conjectured that if there is a polynomially faster than $O(mn)$ time algorithm for this problem, then one can use it to construct a fast algorithm that computes the diameter exactly. These conjectures remain open, but Theorem~\ref{thm:hard} shows that the $2$ vs $3$ diameter problem may be hard to solve very efficiently for a different reason.

%Our hardness result focuses on the case when the input graph has constant diameter ($2$ or $3$).
%Notice that approximating the diameter is more challenging when the diameter is small: when the input graph has diameter $D\geq n^\eps$ for some $\eps>0$, one can find an arbitrarily good approximation by random sampling: if you randomly sample $C n^{1-\eps}/\delta \log n$ nodes, then with probability at least $1-1/n^C$, one of these nodes is at distance at least $(1-\delta)D$ from an endpoint of the diameter path; hence a $1/(1-\delta)$-approximation can be found in $\tilde{O}(mn^{1-\eps}/\delta)$ time by BFS.
%For sparse enough graphs of diameter $n^{o(1)}$ however, the best known $(3/2-\eps)$-approximation algorithms still compute the diameter exactly in $\tilde{O}(mn)$ time.

Theorem~\ref{thm:hard} shows that unless SETH fails, the best one can do with an $O(m^{2-\eps})$ time algorithm is a $3/2$-approximation.
The Aingworth \etal $3/2$-approximation algorithm almost achieves an $O(m^{2-\eps})$ runtime, except for very sparse graphs when it still runs in $\Omega(n^2)$ time. We notice that with a slight change in the parameters of the algorithm, the  Aingworth \etal running time can be modified to be $\Ot(m^{2/3}n)\leq \Ot(m^{2-1/3})$.
We then investigate whether we can obtain a $3/2$-approximation algorithm that improves upon these two runtimes of the Aingworth \etal algorithm.
% for sparse graphs.
%it is possible to improve on the $\tilde{O}(n^2+m\sqrt{n})$ time
%$3/2$-approximation algorithm of Aingworth \etal~\cite{aingworth}. One drawback of the Aingworth \etal algorithm is that it runs in $\Omega(n^2)$ time even in sparse graphs.
We give a new $3/2$-approximation algorithm with $\tilde{O}(m\sqrt{n})$ expected running time, thus removing the $n^2$ additive factor from the original Aingworth \etal runtime with some randomization, and also beating $\Ot(m^{2/3}n)$.
Our algorithm is the first improvement over the Aingworth \etal diameter algorithm. The improvement is especially noticeable for sparse graphs (with $m=\Ot(n)$) in which our algorithm
 runs in $\tilde{O}(n^{1.5})$ time. Previously, such a result was known only for sparse {\em planar} graphs~\cite{BeKa07}\footnote{disregarding polylogarithmic factors}.
  We also show that in some special cases our algorithm obtains an approximation that is better than $3/2$.

%
%We prove the following Theorem in Section~\ref{s-improve-runtime}
\begin{theorem}
Let $G=(V,E)$ be a directed or an undirected graph with diameter $D=3h+z$, where $h\geq 0$ and $z\in \{0,1,2\}$. In $\Ot(m \sqrt n)$ expected time one can compute an estimate $\hat{D}$ of $D$ such that $2h+z \leq \hat{D} \leq D$ for $z\in \{0,1\}$ and $2h+1 \leq \hat{D} \leq D$ for $z=2$.\label{thm:improved}
\end{theorem}

For undirected or directed graphs with arbitrary nonnegative weights, we also obtain the following.

\begin{theorem}
Let $G=(V,E)$ be a directed or an undirected graph with nonnegative edge weights and diameter $D$. In $\Ot(m \sqrt n)$ expected time one can compute an estimate $\hat{D}$ of $D$ such that $\lfloor 2D/3\rfloor\leq \hat{D} \leq D$.\label{thm:improvedw}
\end{theorem}

We further investigate whether one can improve the approximation for unweighted graphs obtained in Theorem~\ref{thm:improved} by possibly increasing the runtime, while still keeping it subcubic in $n$.
%how well one can do in unweighted graphs.
% in terms of approximation.
%We show that if one is willing to spend more time it is possible to get a better approximation.
Notice that in Theorem~\ref{thm:improved}, the estimate $\hat{D}$ is at least $2h+z$ for $z\in \{0,1\}$ and only at least $2h+1$ for $z=2$. 
This only guarantees that $\hat{D}\geq \lfloor 2D/3\rfloor$. (This is also the case for the algorithm of Aingworth \etal~\cite{aingworth}.) 

 We show that with a slightly larger (but still subcubic) running time it is possible to get an estimate $\hat{D}$ of $D$ such that $2h+z \leq \hat{D}$ for any value $z\in \{0,1,2\}$, thus guaranteeing that $\hat{D}\geq \lceil 2D/3\rceil$. This is significant when $D$ is a constant, and also shows that when $z\neq 0$, the approximation factor is strictly better than $3/2$: $(3h+z)/(2h+z) = 3/2 - 1/(4h/z+2)\leq 3/2 - 1/(4h+2) < 3/2.$

%Our hardness result focuses on the case when the input graph has constant diameter ($2$ or $3$).
We note that approximating the diameter is most challenging when the diameter is small. When the input graph has diameter $D\geq n^\eps$ for some $\eps>0$, one can efficiently find an arbitrarily good approximation by random sampling: if you randomly sample $C n^{1-\eps}/\delta \log n$ nodes, then with probability at least $1-1/n^C$, one of these nodes is at distance at least $(1-\delta)D$ from an endpoint of the diameter path; hence a $1/(1-\delta)$-approximation can be found in $\tilde{O}(mn^{1-\eps}/\delta)$ time by BFS.
For sparse enough graphs of diameter $n^{o(1)}$ however, the best known $(3/2-\eps)$-approximation algorithms still compute the diameter exactly in $\tilde{O}(mn)$ time. Hence, it is quite interesting that we can obtain $\tilde{O}(m\sqrt n)$ time $(3/2-\eps)$-approximation algorithms for some constant values of the diameter.

In Section~\ref{s-improve-approx} we prove the following Theorem.

\begin{theorem}\label{T-better-approx}
Let $G=(V,E)$ be a directed or undirected unweighted graph with diameter $D=3h+z$, where $h\geq 0$ and $z\in \{0,1,2\}$.
There is an $\tilde{O}(m^{2/3} n^{4/3})$ time algorithm that reports an estimate $\hat{D}$ such that $2h+z \leq \hat{D}\leq D$.
\end{theorem}

Marginally, we show how to get a better estimate for undirected graphs in the same running time.
\begin{theorem}\label{undir-approx}
Let $G=(V,E)$ be an undirected unweighted graph with diameter $D$.
There is an $\tilde{O}(m^{2/3} n^{4/3})$ time algorithm that reports an estimate $\hat{D}$ such that $\lfloor 4D/5\rfloor \leq \hat{D}\leq D$.
\end{theorem}

%\begin{theorem}
%There is an $\tilde{O}(m^{2/3} n^{4/3})$ time algorithm that returns an estimate $E$ for the diameter $D$ of directed or undirected unweighted graph,
%%with nonnegative weights,
%so that $E\leq D$. With high probability, $\hat{D}\geq 2h+z$ if $D=3h+z$ and $z\in\{0,1,2\}$, and hence $\hat{D}\geq \lceil 2D/3\rceil$.
%\end{theorem}
%\lceil 2D/3 \rceil = \lceil 2(3h+z)/3 \rceil = 2h+\lceil 2z/3\rceil = 2h+z. this would be true if we could get z=2 also.

The running time in Theorem~\ref{T-better-approx} however is $\tilde{\Theta}(n^2)$ for sparse graphs. We hence investigate whether one can get an estimate $\lceil 2D/3\rceil \leq \hat{D}\leq D$ in $O(m^{2-\eps})$ time. We show:

\begin{theorem}\label{thm:sparsebetter}
There is an $\tilde{O}(m^{2-1/(2h+3)})$ time deterministic algorithm that computes an estimate $\hat{D}$ with $\lceil 2D/3\rceil \leq \hat{D}\leq D$ for all $m$-edge unweighted graphs of diameter $D=3h+z$ with $h\geq 0$ and $z\in \{0,1,2\}$. In particular, $\hat{D}\geq 2h+z$.
\end{theorem}

\paragraph{Notation.}
Let $G=(V,E)$ denote a graph. It can be directed or undirected; this will be specified in each context.
If the graph is weighted, then there is a function on the edges $w:~E\rightarrow \mathbb{Q}^{+}\cup \{0\}$. Unless explicitly specified, the graphs we consider are unweighted.

For any $u,v \in V$, let $d(u,v)$ denote the distance from $u$ to $v$ in $G$.
Let $BFS^{\IN}(v)$ and $BFS^{\OUT}(v)$ be the incoming and outgoing breadth-first search (BFS) trees of $v$, respectively, that is the BFS trees in $G$ starting at $v$ and in $G$ with the edges reversed starting at $v$.
%, that is, all shortest paths in the tree ending (starting) at $v$.
Let $d^{\IN}(v)$ be the depth of $BFS^{\IN}(v)$, i.e. the largest distance from a vertex of $BFS^{\IN}(v)$ to $v$. Similarly, let $d^{\OUT}(v)$ be the depth of $BFS^{\OUT}(v)$.

For $h\leq d^{\IN}(v)$, let $BFS^{\IN}(v,h)$ be the vertices in the first $h$ levels of $BFS^{\IN}(v)$. Similarly, for $h\leq d^{\OUT}(v)$, let $BFS^{\OUT}(v,h)$ be the vertices in the first $h$ levels of $BFS^{\OUT}(v)$.

Let $N_s^{\IN}(v)$ ($N_s^{\OUT}(v)$) be the set of the $s$ closest incoming (outgoing) vertices of $v$, where ties are broken by taking the vertex with the smaller id. We assume throughout the paper that for each $v$ and each $s\leq n$, $|N_s^{\IN}(v)|=|N_s^{\OUT}(v)|=s,$ as otherwise the diameter of the graph would be $\infty$, and this can be checked with two BFS runs from and to an arbitrary node.

Let $d_s^{\IN}(v)$ be the largest distance from a vertex of $N_s^{\IN}(v)$ to $v$, and
$d_s^{\OUT}(v)$ be the largest distance from $v$ to a vertex of $N_s^{\OUT}(v)$.
Let $d_s^{\IN} = \max_{v\in V} d_s^{\IN}(v)$ and $d_s^{\OUT} = \max_{v\in V} d_s^{\OUT}(v)$.

For a set $S\subseteq V$ and a vertex $v\in V$ we define $p_S(v)$ to be a vertex of $S$ such that $d(v,p_S(v))\leq d(v,w)$ for every $w \in S$, i.e. the closest vertex of $S$ to $v$.

For a degree $\Delta$ we define $p_\Delta(v)$ to be the closest vertex to $v$ of degree at least $\Delta$, that is, $d(v,p_\Delta(v))\leq d(v,w)$ for every $w \in V$ of degree at least $\Delta$.

We use the following standard notation for running times. For a function of $n$, $f(n)$, $\tilde{O}(f(n))$ denotes $O(f(n)\poly\log n)$ and $O^{*}(f(n))$ denotes $O(f(n)\poly~ n)$.

%\paragraph{Paper organization.}
%The rest of the paper is organized as follows. In Section~\ref{s-lowerbound} we present our hardness result. In Section~\ref{s-aingworth} we present the algorithm of Aingworth \etal~\cite{aingworth}. In Section~\ref{s-improve-runtime} we present our faster algorithm. In Section~\ref{s-improve-approx} we present our algorithms with the better approximation. We end in Section~\ref{s-conc} with concluding remarks and open problems.
%
\section{Diameter approximation and the Strong Exponential Time Hypothesis}\label{s-lowerbound}

Impagliazzo, Paturi, and Zane~\cite{ipz1,ipz2} introduced the Exponential Time
Hypothesis (ETH) and its stronger variant, the Strong Exponential Time
Hypothesis (SETH). These two complexity hypotheses assume lower bounds on
how fast satisfiability problems can be solved. They have frequently been
used as a basis for conditional lower bounds for other concrete computational
problems.

%R. Impagliazzo and R. Paturi, On the complexity of k-sat, J. Comput.
%Syst. Sci., 62 (2001), pp. 367–375.
%R. Impagliazzo, R. Paturi, and F. Zane, Which problems have strongly
%exponential complexity?, Journal of Computer and System Sciences, 63 (2001),
%pp. 512–530.

\begin{hyp}[\cite{ipz1,ipz2}]
{\bf ETH}: There exists a real constant $\delta>0$ such that $3$-SAT instances on $n$ variables and $m$ clauses cannot be solved in $2^{\delta n}\poly(m,n)$ time.\end{hyp}

A natural question is how fast can one solve $r$-SAT as $r$ grows.
Impagliazzo, Paturi, and Zane define $$s_r=\inf \{\delta~|~\exists~ O^{*}(2^{\delta n}) \textrm{ time algorithm solving } r\textrm{-SAT instances with } n \textrm{ variables}\}, \textrm{ and } s_\infty = \lim_{r\rightarrow \infty} s_r.$$
Clearly $s_r\leq s_{r+1}$ so that the sequence is nondecreasing. Impagliazzo, Paturi, and Zane show that if ETH holds, then $s_r$ also increases infinitely often. Furthermore, all known algorithms for $r$-SAT nowadays take time $O(2^{n(1-c/r)})$ for some constant $c$ independent of $n$ and $r$ (e.g.~\cite{hirschsat,moniensat,PaturiPZ99,PPSZ05,Sch92sat,Scho99sat}).
Because of this, it seems plausible that $s_\infty=1$, and this is exactly the strong exponential time hypothesis.

\begin{hyp}[\cite{ipz1,ipz2}]
{\bf SETH}: $s_\infty=1$.
\end{hyp}

One immediate consequence of SETH is that CNF-SAT on $n$ variables cannot be solved in $2^{n(1-\eps)}\poly(n)$ time for any $\eps>0$.
%This consequence is believable since t
The best known algorithm for CNF-SAT is the $O^{*}(2^n)$ time  exhaustive search algorithm which tries all possible $2^n$ assignments to the variables, and it has been a major open problem to obtain an improvement.
Cygan et al.~\cite{cygan} showed that SETH is also equivalent to the assumption that several other NP-hard problems cannot be solved faster than by exhaustive search, and the best algorithms for these problems are the exhaustive search ones.
%This also gives some credibility to SETH.
%
%On Problems as Hard as CNFSAT
%Marek Cygan (University of Lugano), Holger Dell (University of Wisconsin-Madison), Daniel Lokshtanov (University of California, San Diego), Daniel Marx (Hungarian Academy of Sciences), Jesper Nederlof (Utrecht University), Yoshio Okamoto, (Japan Advanced Institute of Science and Technology), Ramamohan Paturi (University of California, San Diego), Saket Saurabh (Institute of Mathematical Sciences, India), Magnus Wahlstrom (Max-Planck-Institut fur Informatik)

Assuming SETH, one can prove tight conditional lower bounds on the complexity of some problems in P as well. The problem that we will look at is $k$-dominating set for constant $k$: given an undirected graph $G=(V,E)$, is there a set $S$ of $k$ vertices so that every vertex $v\in V$ is either in $S$ or has an edge to some vertex in $S$?
The best known algorithm for $k$-dominating set for $k\geq 7$ runs in $n^{k+o(1)}$ time and uses rectangular matrix multiplication~\cite{PW10}. P\v{a}tra\c{s}cu and Williams~\cite{PW10} showed that improving on this runtime may be hard as it would imply faster algorithms for CNF-SAT.

\begin{theorem}[\cite{PW10}]
Suppose there is a $k \geq 3$ and function $f$ such that $k$-Dominating Set in an $N$-node graph is
in $O(N^{f(k)})$ time. Then CNF-SAT on $n$ variables and $m$ clauses is in $O^{*}( (m + k 2^{n/k})^{f(k)} )$ time.
\end{theorem}

If $f(k)=k-\varepsilon$ for some constant $\varepsilon>0$, then the above implies that SETH is false.

%The best algorithm for CNF-SAT on $n$ variables and $m$ clauses runs in $2^{n-n/O(\log (m/n)) \poly m}$ time (Calabro, Paturi, and Impagliazzo and observed by Dantsin and Hirsch).
%To improve on this, we need $f(k)/k < 1-1/(c\log m/n)$ for all values of $m$.
%If $f(k)=k-\varepsilon$, then we need
%$\varepsilon/k > 1/(c\log m/n)$ for all values of $m$. In particular $\varepsilon > k/(c' n)$ is sufficient. Since the number of nodes $N$ in the reduction is roughly $2^{n/k}$, $k/n$ is roughly $1/\log N$ and so $\varepsilon > 1/(c''\log N)$ is what we need.

%Chris Calabro, Russell Impagliazzo, and Ramamohan Paturi. A Duality between ClauseWidth
%nd Clause Density for SAT. In Proc. IEEE Conference on Computational Complexity, 252–
%260, 2006
%Evgeny Dantsin and Edward A. Hirsch. Worst-Case Upper Bounds. In Handbook of Satisfiability, Armin Biere, Marijn Heule, Hans van Maaren and Toby Walsch (eds.), 341–362, 2008.

We show a strong relationship between the diameter problem in undirected unweighted graphs and $k$-dominating set.

\begin{theorem}Suppose one can distinguish between diameter $2$ and $3$ in an $m$-edge undirected unweighted graph in time $O(m^{2-\varepsilon})$ for some constant $\varepsilon>0$.
Then
%there is a constant $\delta>0$ such that
 for all integers $k\geq 2/\varepsilon$, $2k$-dominating set can be solved in $O^{*}(n^{2k-\varepsilon})$ time.
Moreover, CNF-SAT on $n$ variables and $m$ clauses is in $O^{*}(2^{n(1-\varepsilon^2/4)})$ time, and SETH is false.\label{thm:hard2}
\end{theorem}

Theorem~\ref{thm:hard2} immediately implies Theorem~\ref{thm:hard} in the introduction, as any $(3/2-\eps)$-approximation algorithm can distinguish between diameter $2$ and $3$.
%Note also that the proof shows a stronger thing: that diameter can solve dominating set!
%*** consequences of k-dominating set having a good algorithm?

%if k-dom set has an $n^{o(k)}$ time algorithm, then FPT=W[1] (Chen et al 04).
%useless since the algorithm requires atleast linear time.

%*** consequences of CNF-SAT improved alg:
%%The best algorithm for CNF-SAT on $n$ variables and $m$ clauses runs in $2^{n-n/O(\log (m/n))} \poly m$ time (Calabro, Paturi, and Impagliazzo and observed by Dantsin and Hirsch).
%%We would improve on this if $\eps^2>> 1/O(\log (m/n))$ for all values of $m$.

%To improve on this, we need $f(k)/k < 1-1/(c\log m/n)$ for all values of $m$.
%If $f(k)=k-\varepsilon$, then we need
%$\varepsilon/k > 1/(c\log m/n)$ for all values of $m$. In particular $\varepsilon > k/(c' n)$ is sufficient. Since the number of nodes $N$ in the reduction is roughly $2^{n/k}$, $k/n$ is roughly $1/\log N$ and so $\varepsilon > 1/(c''\log N)$ is what we need.

\begin{proof}
%Assume that there is such an algorithm for some constant $\varepsilon>0$.
%
Given an instance $G=(V,E)$ of $2k$-Dominating set for constant $k$, we construct an instance of the $2$ vs $3$ diameter problem and we show that $2k$-Dominating set in $n$-node graphs can be solved in $O^{*}(n^{2k-\delta})$ time for some constant $\delta>0$ depending on $\varepsilon$.

Take all $k$-subsets of the vertices in $V$ and add a node for each of them to the $2$ vs $3$ instance $G'$.
Add a node for every vertex in $V$ -- call this set of nodes $V'$ and make $V'$ into a clique.

For every $k$-subset $S$ of vertices of $V$, connect $S$ to $v\in V'$ in $G'$ iff $S$ \emph{does not dominate} $v$ in $G$.
While we do this we check whether each $S$ is a $k$-dominating set in $G$, and if so, we stop. From now on we can assume that none of the $k$-subsets $S$ are dominating sets in $G$.

Now, notice that if $S$ and $T$ are two $k$-subsets so that their union is not a $(\leq 2k)$-dominating set in $G$, then the distance in $G'$ between $S$ and $T$ is $2$: there is some $u$ that is dominated by neither $S$ nor $T$ and so $S-u-T$ is a path of length $2$.
If, on the other hand, $S\cup T$ is a dominating set in $G$, then there is no such path and the shortest path between $S$ and $T$ in $G'$ is to go from $S$ to some $v$ that $S$ doesn't dominate, then to some $u$ that $T$ doesn't dominate ($V'$ is a clique) and then from $u$ to $T$.
%(The nodes $u$ and $v$ exist since otherwise $S$ or $T$ is a $k$-dominating set in $G$.)

The distance between any $u$ and $v$ in $V'$ is $1$, and the distance between any $u$ and any $S$ is at most $2$:
%(unless $S$ is a $k$-dominating set and would have been detected):
go from $u$ to some node $v$ that $S$ doesn't dominate and then to $S$.

Hence, if there is no $2k$-dominating set in $G$, then the diameter of $G'$ is $2$, and if there is one, then the diameter of $G'$ is $3$.
$G'$ has ${n\choose k} + n$ nodes and at most $O(n\cdot {n\choose k}) \leq O(n^{k+1})$ edges.

Since we can solve the diameter problem in $O(m^{2-\varepsilon})$ time, applying that algorithm to $G'$ solves $2k$-dominating set in $G$ for any $k\geq 2$ in time
$O(n^{2k+2-\varepsilon k -\varepsilon})$.

We want this to be $O(n^{2k-\delta})$ for some $\delta>0$, so it suffices to pick $k$ so that $-\delta\geq 2-\varepsilon (k+1)$.
If we want $\delta=\varepsilon$, then $k\geq 2/\varepsilon$ suffices.
%That is, we want to set $k>2/\varepsilon-1$, so setting $k=\lceil 2/\varepsilon\rceil$ suffices. Then $\delta=\varepsilon$.
\end{proof}

\section{The algorithm of Aingworth \etal}\label{s-aingworth}

In this section we revisit the algorithm of
Aingworth, Chekuri, Indyk and Motwani~\cite{aingworth},
that computes a $3/2$-approximation of the diameter of a directed (or undirected) graph in $\Ot(m\sqrt n + n^2)$ time.
(The algorithm can also be made to work for graphs with nonnegative weights with roughly the same running time and approximation factor. In this section we only focus on the algorithm for unweighted graphs.)
% We will revisit the case of weighted graphs later on.)

Let $s$ be a given parameter in $[1,n]$.
The algorithm works as follows.
First, it computes $N_s^\OUT(v)$ for every $v\in V$. Then, for a vertex $w$, where $d_s^{\OUT}(w) = d_s^{\OUT}$ it computes
$BFS^{\OUT}(w)$ and for every $u\in N_s^\OUT(w)$ it computes $BFS^{\IN}(u)$.
Next, it computes a set $S$ that hits $N_s^\OUT(v)$ for every $v\in V$ and for every $u\in S$ it computes $BFS^{\OUT}(u)$.
As an estimate, the algorithm returns the depth of the deepest computed BFS tree.

%In the next Lemma we analyze the running time of the algorithm.
The next lemma appears in~\cite{aingworth}. We state it for completeness. 
%Aingworth et al. set $s=\sqrt{n}$ and obtain their running time.
\begin{lemma}\label{L-Aing-runtime}
The running time of the algorithm is $\Ot(ns^2+(n/s+s)m)$.
\end{lemma}
%\begin{proof}
%The cost of computing $N_s^\OUT(v)$ for every $v\in V$ is $O(ns^2)$. A hitting set $S$ of size $\Ot(n/s)$ can be computed deterministically in $O(m+ns)$ time. Finally, the time for computing the BFS trees for the vertices of $S \cup N_s^\OUT(w)$ is $\Ot(m(n/s+s))$.
%\end{proof}

Aingworth et al. set $s=\sqrt{n}$ and obtain their running time. We note that if one sets $s= m^{1/3}$ instead, one can get a runtime of $\Ot(m^{2/3}n)$ that is better for sparse graphs; we later show that both of these runtimes can be improved with randomization.

We now analyze the quality of the estimate returned by the algorithm. Aingworth~\etal\cite{aingworth} proved that this estimate %returned by their algorithm for a graph of diameter $D$
is at least $\lfloor 2D/3 \rfloor$ in graphs of diameter $D$. Here we present a tighter analysis.

\begin{lemma}\label{L-Aing-approx}
Let $G=(V,E)$ be a directed graph with diameter $D=3h+z$, where $h\geq 0$ and $z\in \{0,1,2\}$.
Let $\hat{D}$ be the estimate returned by the algorithm. For $z\in \{0,1\}$, we have $2h+z \leq \hat{D} \leq D$. For $z=2$, we have that $2h+1 \leq \hat{D} \leq D$.
\end{lemma}
\begin{proof}
Let $a,b\in V$ such that $d(a,b)=D$.
First notice that the algorithm always returns a depth of some shortest paths tree and hence  $\hat{D} \leq D$.

Now, if $d_s^{\OUT}(w)\leq h$ then also $d_s^{\OUT}(a)\leq h$ and as $S$ hits $N_s^\OUT(a)$, one of the BFS trees computed for vertices of $S$ has depth at least $2h+z$. Hence, assume that $d_s^{\OUT}(w)> h$.
We can also assume that $d^{\OUT}(w)<2h+z$ as otherwise when we compute $BFS^\OUT(w)$, we'd return a depth at least $2h+z$.

As $d^{\OUT}(w)<2h+z$, also $d(w,b)<2h+z$. Since $d_s^{\OUT}(w)> h$, we have that $BFS^{\OUT}(w,h)\subseteq N_s^\OUT(w)$. Hence there is a vertex $w'\in N_s^\OUT(w)$ on the path from $w$ to $b$ such that $d(w,w')=h$ and hence
$d(w',b)<h+z$. Since $d(a,b)=3h+z$, we must have that $d(a,w')\geq 2h+1$.
As the algorithm computes $BFS^{\IN}(u)$ for every $u\in N_s^\OUT(w)$, in particular, it computes $BFS^{\IN}(w')$, and returns an estimate $\geq 2h+1$. %Now as $d(w',b)<h+z$, we have that
For $z\in \{0,1\}$, $d(a,w')\geq 2h+1\geq 2h+z$ and hence the final estimate returned is always at least $2h+z$. For $z=2$ we only have that $d(a,w')\geq 2h+1$ and if the algorithm returns $d(a,w')$ as an estimate, it may return $2h+1$ instead of $2h+z$.
\end{proof}

%We end this section with a theorem that summarizes the properties of Aingworth~\etal algorithm.
%
%\begin{theorem}[Aingworth~\etal]\label{T-Aing}
%Let $G=(V,E)$ be a directed graph with diameter $D=3h+z$, where $h\geq 0$ and $z\in \{0,1,2\}$. In $\Ot(n^2+m\sqrt n)$ time one can compute an estimate $\hat{D}$ of $D$ such that $2h+z \leq \hat{D} \leq D$ for $z\in \{0,1\}$ and $2h+1 \leq \hat{D} \leq D$ for $z=2$.
%\end{theorem}

\section{Improving the running time}\label{s-improve-runtime}

The algorithm of Aingworth \etal~\cite{aingworth} runs in $\Ot(ns^2+(n/s+s)m)$. 
%Choosing $s$ to be $\sqrt n$, we get a running time of $\Ot(m\sqrt n + n^2)$. 
In this section we show that it is possible to get rid of the $ns^2$ term, while keeping the quality of the estimate unchanged. 
By choosing $s=\sqrt n$, we get an algorithm running in $\Ot(m\sqrt n)$ time.

The term of $ns^2$ in the running time comes from the computation of $N_s^\OUT(v)$ for every $v\in V$. This computation is done to accomplish two tasks.
One task is to obtain $d_s^{\OUT}(v)$ for every $v\in V$ and then to use it to find a vertex $w$ such that $d_s^{\OUT}(w)=d_s^{\OUT}$.
A second task is to obtain, deterministically, a hitting set $S$ of size $\Ot(n/s)$ that hits the set $N_s^\OUT(v)$ of every $v\in V$.

Our main idea is to accomplish these two tasks without explicitly computing $N_s^\OUT(v)$ for every $v\in V$. The major step in our approach is to completely modify the first task above by picking a different type of vertex to play the role of $w$. Making the second task above fast can be accomplished easily with randomization. We elaborate on this below. 

Our algorithm works as follows. First, it computes a hitting set by using randomization, that is,
%it creates a set $S$ of size $O(n/s \log n)$ by adding every $v\in V$ to $S$ with probability $\Ot(s^{-1})$.
it picks a random sample $S$ of the vertices of size $\Theta(n/s \log n)$.
This guarantees that with high probability (at least $1-n^{-c}$, for some constant $c$), $S \cap N_s^\OUT(v)\neq \emptyset$, for every $v\in V$.
This accomplishes the second task above in $\tilde{O}(n)$ time, with high probability.
 Similarly to the algorithm of  Aingworth \etal~\cite{aingworth}, our algorithm computes $BFS^{\OUT}(v)$, for every $v\in S$.

We now explain the main idea of our algorithm, i.e. how we are able replace the first task from before with a much faster step.
First, for every $v\in V$ our algorithm computes the closest node of $S$, $p_S(v)$, to $v$, by creating a new graph as follows. It adds an additional vertex $r$ with edges $(u,r)$, for every $u\in S$. It computes $BFS^{\IN}(r)$ in this graph. It is easy to see that for every $v\in V$ the last vertex before $r$ on the shortest path from $v$ to $r$ is $p_S(v)$. This step takes $O(m)$ time.

Now, as opposed to the algorithm of  Aingworth \etal that picks a vertex $w$ such that $d_s^{\OUT}(w)=d_s^{\OUT}$, our algorithm finds a vertex $w\in V$
that is furthest away from $S$: i.e. such that
$d(w,p_S(w))\geq d(u,p_S(u))$, for every $u\in V$.
The vertex $w$ plays the same role as its counterpart in~\cite{aingworth}:
Our algorithm computes $BFS^\OUT(w)$ and obtains $N_s^\OUT(w)$ from it. Finally, it computes  $BFS^\IN(u)$ for every $u\in N_s^\OUT(w)$. As an estimate, the algorithm returns the depth of the deepest BFS tree that it has computed.

In the next Lemma we analyze the running time of the algorithm.

\begin{lemma}\label{L-Improve-runtime}
The running time of the algorithm is $\Ot((n/s+s)m)$.
\end{lemma}
\begin{proof}
A hitting set $S$ is formed in $O(n)$ time. With a single BFS computation, in $O(m)$ time, we find $p_S(v)$ for every $v\in V$, and hence also find $w$.
%Finding a vertex $w\in V$ with $d(w,p_S(w))\geq d(u,p_S(u))$, for every $u\in V$, takes $O(n)$ time.
The cost of computing a BFS tree for every $v\in S \cup N_s^\OUT(w)$ is $\Ot((n/s+s)m)$.
\end{proof}

Next, we show that the estimate produced by our algorithm is of the same quality as the estimate produced by Aingworth \etal algorithm, with high probability.

\begin{lemma}\label{L-Improve-approx}
Let $G=(V,E)$ be a directed (or undirected) graph with diameter $D=3h+z$, where $h\geq 0$ and $z\in \{0,1,2\}$.
Let $\hat{D}$ be the estimate returned by the above algorithm. With high probability, $2h+z \leq \hat{D} \leq D$ whenever
$z\in \{0,1\}$, and $2h+1 \leq \hat{D} \leq D$ whenever $z=2$.
\end{lemma}
\begin{proof}
Let $a,b\in V$ such that $d(a,b)=D$. Let $w$ be a vertex that satisfies $d(w,p_S(w))\geq d(u,p_S(u))$, for every $u\in V$.

If $d(w,p_S(w))\leq h$ then also $d(a,p_S(a))\leq h$.
% since $d(w,p_S(w))\geq d(u,p_S(u))$, for every $u\in V$.
 As the algorithm computes $BFS^\OUT(v)$ for every $v\in S$, it follows that $BFS^\OUT(p_S(a))$ is computed as well and its depth is at least $2h+z$ as required.  Hence, assume that $d(w,p_S(w)) > h$. We can assume also that $d^{\OUT}(w)<2h+z$ since the algorithm computes $BFS^\OUT(w)$ and if $d^{\OUT}(w)\geq 2h+z$ then it computes a BFS tree of depth at least $2h+z$ as required.

Since $d^{\OUT}(w)<2h+z$ it follows that $d(w,b)<2h+z$. Moreover, since $d(w,p_S(w))>h$ and $S$ hits $N_s^\OUT(w)$ whp, we must have that $N_s^\OUT(w)$ contains a node at distance $>h$ from $w$, and hence $BFS^{\OUT}(w,h)\subseteq N_s^\OUT(w)$. This implies that there is a vertex $w'\in N_s^\OUT(w)$ on the path from $w$ to $b$ such that $d(w,w')=h$ and hence $d(w',b)<h+z$. Since $d(a,b)=3h+z$, we also have that $d(a,w')\geq 2h+1$.

The algorithm computes $BFS^{\IN}(u)$ for every $u\in N_s^\OUT(w)$, and in particular, it computes $BFS^{\IN}(w')$, thus returning an estimate at least $d(a,w')\geq 2h+1$.
Hence for $z\in \{0,1\}$ the final estimate is always $\geq 2h+z$, and for $z=2$ the estimate could be $2h+1$ but no less.
%Now as $d(w',b)<h+z$ for $z\in \{0,1\}$ it follows that $d(w',b)\leq h$ and $d(a,w')\geq 2h+z$. For $z=2$ it follows that $d(w',b)\leq h+1$ and $d(a,w')\geq 2h+1$.
\end{proof}

We now turn to prove Theorem~\ref{thm:improved} from the introduction.

\begin{reminder}{Theorem~\ref{thm:improved}}
Let $G=(V,E)$ be a directed or an undirected graph with diameter $D=3h+z$, where $h\geq 0$ and $z\in \{0,1,2\}$. In $\Ot(m \sqrt n)$ expected time one can compute an estimate $\hat{D}$ of $D$ such that $2h+z \leq \hat{D} \leq D$ for $z\in \{0,1\}$ and $2h+1 \leq \hat{D} \leq D$ for $z=2$.
\end{reminder}

%\begin{theorem}\label{T-faster}
%Let $G=(V,E)$ be a directed graph with diameter $D=3h+z$, where $h\geq 1$ and $z\in \{0,1,2\}$. In $\Ot(\sqrt n m)$ expected time it is possible to compute an estimate $\hat{D}$ of $D$ such that $2h+z \leq \hat{D} \leq D$ for $z\in \{0,1\}$ and $2h+1 \leq \hat{D} \leq D$ for $z=2$.
%\end{theorem}
\begin{proof}
From Lemma~\ref{L-Improve-runtime} we have that if we set $s=\sqrt n$ the algorithm runs in $\Ot(m\sqrt n)$ worst case time.
From Lemma~\ref{L-Improve-approx} we have that with high probability, that is $1-n^{-c}$ for some constant $c$, the algorithm returns an estimate of the desired quality. We now show how to convert the algorithm into a Las-Vegas one so that it
always returns an estimate of the desired quality but the
running time is $\Ot(m \sqrt n)$ in expectation.

Randomization is used only in order to obtain a set that hits $N_s^\OUT(v)$ for every $v\in V$.
The only place that the hitting set affects the quality of the approximation is in Lemma~\ref{L-Improve-approx} where we used the fact that, whp,
$S$ contains a node of $N_s^\OUT(w)$, so that if $d(w,S)>h$, $N_s^\OUT(w)$ contains a node at distance $>h$ from $w$.

Note that we compute $N_s^\OUT(w)$ and we can check whether $S$ intersects it in $\Ot(s)$ time. If it doesn't, we can
%
%$BFS^{\OUT}(w,h)\subseteq N_s^\OUT(w)$.
%However, we can check that this is the case while computing $N_s^\OUT(w)$ and if
%$BFS^{\OUT}(w,h)\nsubseteq N_s^\OUT(w)$ then to
rerun the algorithm until we have verified that $S\cap N_s^\OUT(w)\neq \emptyset$.
%$BFS^{\OUT}(w,h)\subseteq N_s^\OUT(w)$.
Since $S\cap N_s^\OUT(w)= \emptyset$ holds with very small probability, the expected running time of the algorithm is $\Ot(m \sqrt n)$ and its estimate is guaranteed to have the required quality.
\end{proof}

Just as in~\cite{aingworth}, we can make our algorithm work for graphs with nonnegative weights as well by replacing every use of BFS with Dijkstra's algorithm. The proofs are analogous, and the running time is increased by at most a $\log n$ factor. We obtain

\begin{reminder}{Theorem~\ref{thm:improvedw}}
%Let $G=(V,E)$ be a directed or an undirected graph with diameter $D$. In $\Ot(m \sqrt n)$ expected time one can compute an estimate $\hat{D}$ of $D$ such that $\lfloor 2D/3\rfloor\leq \hat{D} \leq D$.
Let $G=(V,E)$ be a directed or an undirected graph with nonnegative edge weights and diameter $D$. In $\Ot(m \sqrt n)$ expected time one can compute an estimate $\hat{D}$ of $D$ such that $\lfloor 2D/3\rfloor\leq \hat{D} \leq D$.
\end{reminder}

\section{Improving the approximation for unweighted graphs}\label{s-improve-approx}

In this section we show that in some cases it is possible to improve the approximation of the algorithm of Aingworth \etal for unweighted graphs.
Recall that for a graph with diameter $D=3h+2$ their algorithm returns an estimate $\hat{D}$ such that $2h+1 \leq \hat{D} \leq D$. We show that for such a case it is possible to return an estimate $\hat{D}$ such that $2h+2 \leq \hat{D} \leq D$. This is significant for small diameter values. For example, for a graph of diameter $5$ our estimate is at least $4$, while the previous estimate was at least $3$.

We present two algorithms that obtain this improved approximation, one works well for dense graphs and the other for sparse graphs.

\subsection{Dense graphs}

%\begin{algorithm}[t]\label{A-approx}
%\caption{Approx-Diam($G$)}
%$X_1 \gets$ Aingworth($G$)\;
%$X_2 \gets$ Aingworth($G^R$)\;
%
%$\hat{D} \gets \max (X_1,X_2)$\;
%
%\ForEach{$v\in V$}{\ForEach{$u\in V\setminus \{ v \}$}
%{
%
%\If{$BFS^\OUT(u,d^\OUT_s(u)-1)\cap BFS^\IN(v,d^\IN_s(v)-1)=\emptyset \wedge \nexists (u',v')\in E$ s.t. $u'\in BFS^\OUT(u,d^\OUT_s(u)-1) \wedge v'\in BFS^\IN(v,d^\IN_s(v)-1)$}{$\hat{D} \gets \max (\hat{D},d^\OUT_s(u)+d^\IN_s(v))$}}}
%\Return $\hat{D}$\;
%\end{algorithm}

%Compute sets $S^\OUT$ and $S^\IN$ that hit  every out and in neighborhood of size $s$, respectively\;
%$x^o \gets \arg\max_{x\in V} d_s^{\OUT}(x)$\;
%$x^i \gets \arg\max_{x\in V} d_s^{\IN}(x)$\;
%\lForEach{$ v\in S^\OUT$}{Compute $BFS^\OUT(v)$\;}
%\lForEach{$ v\in S^\IN$}{Compute $BFS^\IN(v)$\;}
%
%$y^o \gets \arg\max_{x\in S^\OUT} d^\OUT(x)$\;
%$y^i \gets \arg\max_{x\in S^\IN} d^\IN(x)$\;
%$h^{\max} \gets  \max (d^\OUT(y^o),d^\IN(y^i))$\;
%
%
%\lIf{$(d_s^\OUT(x^o)\leq h) \bigvee (d_s^\IN(x^i)\leq h)$}{\Return $h^{\max}$\;}
%
%
%$F^\OUT \gets \{ v \mid d_s^\OUT(v)>h\}$\;
%$F^\IN \gets \{ v \mid d_s^\IN(v)>h\}$\;
%
%\ForEach{$(u,v)\in F^\OUT \times F^\IN$ }{\If {$BFS^\OUT(u,h)\cap BFS^\IN(v,h)=\emptyset$ and $\nexists (u',v')\in E$ s.t. $u'\in BFS^\OUT(u,h) \wedge v'\in BFS^\IN(v,h)$}{\Return $\max (2h+2,h^{\max})$} }
%

Our algorithm Approx-Diam($G$) works as follows. (The pseudocode is in the appendix.) First, it runs the Aingworth \etal algorithm both on the input graph $G$ and on the input graph with the edge directions reversed, $G^R$. Let $\hat {D}$ be the maximum value returned by these two runs. A byproduct of this step is that for every $v\in V$ we have computed $BFS^\OUT(v,d^\OUT_s(v)-1)$ and $BFS^\IN(v,d^\IN_s(v)-1)$. Next, our algorithm scans all pairs of vertices $u$ and $v$ and checks whether the following condition holds:
$BFS^\OUT(u,d^\OUT_s(u)-1)$ and $BFS^\IN(v,d^\IN_s(v)-1)$ are disjoint and
 there is no edge between $BFS^\OUT(u,d^\OUT_s(u)-1)$ and $BFS^\IN(v,d^\IN_s(v)-1)$. Given a pair of vertices $u$ and $v$
 for which the condition holds,
% with disjoint $BFS^\OUT(u,d^\OUT_s(u)-1)$ and $BFS^\IN(v,d^\IN_s(v)-1)$, and with no edge between $BFS^\OUT(u,d^\OUT_s(u)-1)$ and $BFS^\IN(v,d^\IN_s(v)-1)$,
 the algorithm updates $\hat {D}$ to be the maximum between its current value and $d^\OUT_s(u)+d^\IN_s(v)$.

%
%
%
%it computes for every vertex its $s$ closest incoming and outgoing vertices. Then, it creates two hitting sets $S^\IN$ and $S^\OUT$ that hit the $s$ closest incoming and outgoing sets of every vertex, respectively.
%The size of both $S^\IN$ and $S^\OUT$ is $\Ot(n/s)$.
%
%The algorithm checks whether the $s$ closest outgoing vertices of every vertex are hit by a vertex at distance at most $h$ from $S^\OUT$. In such a case it returns $d^\OUT(y^o)$, where $y^o\in S^\OUT$ is the vertex with the deepest BFS tree. If not, the algorithm checks whether the $s$ closest incoming vertices of every vertex are hit by a vertex at distance at most $h$ from $S^\IN$. In such a case it returns $d^\IN(y^i)$, where $y^i\in S^\IN$ is the vertex with the deepest BFS tree.
%
%If not then the algorithm searches for a pair of vertices $(u,v)\in F^\OUT \times F^\IN$, where $F^x=\{ v \mid d_s^x(v)>h\}$ and $x\in \{\IN,\OUT\}$, such that
%$BFS^\OUT(u,h) \cap BFS^\IN(v,h)$ is empty and there is no edge between $BFS^\OUT(u,h)$ and $BFS^\IN(v,h)$. For such a pair $(u,v)\in F^\OUT \times F^\IN$ it holds that $d(u,v)\geq 2h+2$. If no such pair exists then the maximum distance between vertices of $F^\OUT \times F^\IN$ is at most $2h+1$.
%In the former case the algorithm returns the maximum between $d^\IN(y^i)$, $d^\OUT(y^o)$ and $2h+2$ and in the later case it returns the  maximum between $d^\IN(y^i)$, $d^\OUT(y^o)$ and the maximum distance between vertices of $F^\OUT \times F^\IN$. The algorithm is presented in Algorithm~\ref{A-approx}.
%

We start by showing that the estimate reported by the algorithm is upper-bounded by the graph diameter.

\begin{lemma}\label{L-upper}
Let $G=(V,E)$  be a graph of diameter $D$. If $\hat{D}=$ Approx-Diam($G$), then $\hat{D}\leq D$.
\end{lemma}
\begin{proof}
If Approx-Diam($G$) returns the value that it gets from one of the runs of Aingworth \etal algorithm then the claim follows from Lemma~\ref{L-Aing-approx}. If the algorithm reports $d^\OUT_s(u)+d^\IN_s(v)$ for some pair of vertices $u,v\in V$ it is because there is no edge from $BFS^\OUT(u,d^\OUT_s(u)-1)$ to $BFS^\IN(v,d^\IN_s(v)-1)$, and no vertex in common between the two trees.
This means that there is no path of length at most $d^\OUT_s(u)+d^\IN_s(v)-1$ from $u$ to $v$, and hence, any path from $u$ to $v$, and in particular the shortest one, is of length at least $d^\OUT_s(u)+d^\IN_s(v)\leq D$ as required.
%If Approx-Diam($G,h$) returns either $d^\OUT(y^o)$ or $d^\IN(y^i)$ then the claim trivially holds. If Approx-Diam($G,h$) returns $2h+2$ then there are two vertices $u,v\in V$ such that $BFS^\OUT(u)$ and $BFS^\IN(v)$ are disjoint and there is no edge from a vertex of $BFS^\OUT(u)$ to a vertex of $BFS^\IN(v)$, thus, the shortest path from $u$ to $v$ is of length at least $2h+2$ and as Approx-Diam($G,h$) returns $2h+2$ it follows that $\hat{D}\leq D$. If Approx-Diam($G,h$) returns $\max_{(u,v)\in F^\OUT \times F^\IN} d(u,v)$ then there is a shortest path in the graph of this length, hence, $\hat{D}\leq D$.
\end{proof}

Next, we lower-bound the estimate reported by the algorithm.

\begin{lemma}\label{L-lower}
Let $G=(V,E)$ be a graph of diameter $D=3h+z$, where $h\geq 1$ and $z\in \{0,1,2\}$.
If $\hat{D}=$ Approx-Diam($G$) then $2h+z \leq \hat{D} \leq 3h+z$.
\end{lemma}
\begin{proof}
Let $a,b\in V$ such that $d(a,b)=D$. Running the algorithm of Aingworth \etal for $G$ and the reverse $G^R$ of $G$ implies that we get
an approximation of $2h+z$ in the following cases.

{\bf Case 1:} [$z \neq 2$]. From Lemma~\ref{L-Aing-approx}, we have that the estimate is at least $2h+z$.

{\bf Case 2:} [$d^\OUT_s(a) \leq h$ or $d^\IN_s(b)\leq h$]. If $d^\OUT_s(a) \leq h$ then the hitting set computed by the Aingworth \etal algorithm contains a vertex at distance at most $h$ from $a$ and hence one of the BFS trees that it computes has depth at least $2h+z$. Running the algorithm on $G^R$ guarantees that the same holds when $d^\IN_s(b)\leq h$.

{\bf Case 3:} [$\exists w\in V$ s.t. $d^\OUT_s(w)\geq h+2$].
In this case let $w$  be the vertex with the largest $d^\OUT_s(w)$ value. The Aingworth \etal algorithm computes $BFS^\OUT(w)$. If $d^\OUT(w)\geq 2h+2$ then the claim holds so assume that $d^\OUT(w)\leq 2h+1$. The algorithm computes $BFS^\IN(v)$ for every $v\in BFS^\OUT(w,h+1)$ and since $d(w,b)\leq 2h+1$ there is a vertex $w'\in BFS^\OUT(w,h+1)$ such that $d(w',b)\leq h$. As the algorithm computes $BFS^\IN(w')$ and $d(a,w')\geq 2h+z$ the claim holds.

For the rest of the proof we assume that the three cases above do not hold, hence, $z=2$, $d^\OUT_s(a) = h+1$ and $d^\IN_s(b)= h+1$. The second part of our algorithm searches for a pair of vertices $u,v\in V$ such that there is no edge from $BFS^\OUT(u,d^\OUT_s(u)-1)$ to $BFS^\IN(v,d^\IN_s(v)-1)$ (and no vertex in common between the two trees). As $D=d(a,b)=3h+2>2h+1$, and $d^\OUT_s(a)-1 = h$ and $d^\IN_s(b)-1= h$, we have that there is no edge from $BFS^\OUT(a,d^\OUT_s(a)-1)$ to $BFS^\IN(b,d^\IN_s(b)-1)$ (and no vertex in common between the two trees). Since the estimate reported by the algorithm is the maximum among values that also include $d^\OUT_s(a)+d^\IN_s(b)=2h+2$, we get that $\hat{D}\geq 2h+2$, as required.
\end{proof}

%We have proven Theorem~\ref{T-better-approx} from the introduction.

\begin{reminder}{Theorem~\ref{T-better-approx}}
Let $G=(V,E)$ be a directed or undirected unweighted graph with diameter $D=3h+z$, where $h\geq 0$ and $z\in \{0,1,2\}$.
There is an $\tilde{O}(m^{2/3} n^{4/3})$ time algorithm that reports an estimate $\hat{D}$ such that $2h+z \leq \hat{D}\leq D$.
\end{reminder}

%\begin{theorem}\label{T-better-approx}
%Let $G=(V,E)$ be a graph with diameter $D=3h+z$, where $h\geq 1$ and $z\in \{0,1,2\}$.
%Algorithm Approx-Diam($G$) reports an estimate $\hat{D}$ such that $2h+z \leq \hat{D}\leq D$ in $\Ot(m^{2/3}n^{4/3})$ time.
%\end{theorem}

\begin{proof}
The bounds on the estimate follow from Lemma~\ref{L-lower} and Lemma~\ref{L-upper}. Running the algorithm of  Aingworth \etal takes $\Ot(m(s+n/s)+ns^2)$ time.  Searching for a pair of vertices $u,v\in V$ such that there is no edge from $BFS^\OUT(u,d^\OUT_s(u)-1)$ to $BFS^\IN(v,d^\IN_s(v)-1)$ takes $O(n^2s^2)$ time. Setting $s=(m/n)^{1/3}$ gives us the running time.
\end{proof}

We can use Theorem~\ref{T-better-approx} to obtain an even better approximation for undirected graphs.

\begin{reminder}
{Theorem~\ref{undir-approx}}
Let $G=(V,E)$ be an undirected unweighted graph with diameter $D$.
There is an $\tilde{O}(m^{2/3} n^{4/3})$ time algorithm that reports an estimate $\hat{D}$ such that $\lfloor 4D/5\rfloor \leq \hat{D}\leq D$.\end{reminder}

%\begin{corollary}\label{C-undirected-graphs}
%Let $G=(V,E)$ be an unweighted undirected graph with diameter $D$.
%There is an $\Ot(m^{2/3}n^{4/3})$ time algorithm that
%reports an estimate $\hat{D}$ such that $\lfloor 4D/5 \rfloor \leq \hat{D}\leq D$.
%\end{corollary}
\begin{proof}
Using~\cite{DorHZ00} we compute the distances between every pair of vertices in the graph, with an additive error of $2$ in $O(\min(n^{3/2} \sqrt m,n^{7/3}))$ time.
If $\hat{D}$ is the maximum distance minus $2$ then $D-2 \leq \hat{D}\leq D$. For every $D\geq 6$ we have that $D-2\geq \lfloor 4/5D \rfloor$. Thus, when $\hat{D}\geq 4$ we get an estimate of at least $\lfloor 4D/5 \rfloor$. If $\hat{D}=3$ then $D$ might be either $3$, $4$ or $5$, that is, $D=3+z$, where $z\in \{0,1,2\}$. If $D=5$, an estimate of $3$ is not good enough, thus we run Approx-Diam($G$). Let $D'$ be the estimate reported by Approx-Diam($G$). From Lemma~\ref{L-lower} it follows that if $D=5$ then $D'\geq 4$ and we have the required approximation. If $\hat{D}=2$ then $D$ might be either $2$, $3$ or $4$, and for this case we can just use the Aingworth \etal algorithm to get an estimate of $3$ whenever $D=4$ which gives the desired approximation.
\end{proof}

%
%We can now use these lemmas to distinguish between graphs of different diameters.
%
%\begin{theorem}\label{T-decision}
%Let $h>0$ be an integer. Let $G=(V,E)$ be a graph whose diameter $D$ is either $3h+z$ or $2h+z-1$, where $z\in \{0,1,2\}$.  Using Approx-Diam($G,h$) we can decide whether $D=3h+z$ or $D=2h+z-1$.
%\end{theorem}
%\begin{proof}
%Let $\hat{D}=$Approx-Diam($G,h$). From Lemma~\ref{L-lower} it follows that if $D=3h+z$ then $\hat{D}\geq 2h+z$. From Lemma~\ref{L-upper} it follows that if $D=2h+z-1$ then $\hat{D}< 2h+z$.
%\end{proof}

\subsection{Sparse graphs}

We now show that it is possible to obtain the better approximation also in $\Ot(m^{2-\eps})$ time for constant $\eps>0$ when the diameter of the given graph is constant.

Our algorithm, Approx-Diam-Sparse($G,\tilde{h})$ is given an estimate $\tilde{h}$ of $h$ so that 
%$h\leq \tilde{h}\leq 2h+1$ 
$\tilde{h}\geq h$ and works as follows. (The pseudocode can be found in the appendix.) Let $\Delta$ be a parameter and let $H$ be the set of vertices of outdegree at least $\Delta$.
 For every vertex of $H$, the algorithm computes an outgoing BFS tree.
 Then, it computes the distance from every node in $V\setminus H$ to $H$. This is done by adding an extra node $r$ to the graph with edges from each node of $H$ to $r$ and then computing an incoming BFS to $r$ in $O(m)$ time. The distance of a node $v$ to $H$ is its distance to $r$, minus $1$.
 The algorithm then picks the vertex $w$ that is furthest from $H$ and computes $BFS^\OUT(w)$.
 Let $h'=\min\{\tilde{h}+1, d(w,H)\}$. The algorithm computes $BFS^\IN(v)$ for every $v\in BFS^\OUT(w,h')$.
% If $d(w,H)\geq h+1$, then the algorithm computes $BFS^\OUT(w)$ and for each vertex $v\in BFS^\OUT(w,h+1)$ it computes $BFS^\IN(v)$.
 %arbitrary vertex $w$ whose closest vertex of degree at least $\Delta$ is at a distance of at least $h+1$, if such a $w$ exists.
 Finally, it returns the maximum depth of all computed BFS trees.
% We give pseudo-code in Algorithm~\ref{A-approx-m} for the more general case when $h$ may not be known, but instead an estimate of $h$, $\tilde{h}$ is given.

We now analyze the quality of the approximation.

\begin{lemma}\label{L-estimate-sparse}
Let $G=(V,E)$ be a graph of constant diameter $D=3h+z$, where $h\geq 0$ and $z\in \{0,1,2\}$.
If $\hat{D}=$ Approx-Diam-Sparse($G,\tilde{h}$) for $\tilde{h}\geq h$, then $2h+z \leq \hat{D} \leq D$.
\end{lemma}
\begin{proof}
First notice, that in any case the algorithm returns a depth of some BFS tree in the graph, thus $\hat{D} \leq D$.

Now, let $a,b\in V$ such that $d(a,b)=D$ and let $H\subseteq V$ be the set of vertices of outdegree at least $\Delta$. Let $y^o\in H$ be the vertex with the deepest outgoing BFS in $H$.
%Let $w$ be an arbitrary vertex that satisfies $BFS^\OUT(w,h) \cap H = \emptyset$, if one exists.
Let $y^i$ be the vertex with the deepest incoming BFS among the vertices of $BFS^\OUT(w,h')$, where $h'=\min\{\tilde{h}+1,d(w,H)\}$.
The algorithm returns as an estimate $\max(d^\OUT(y^o),d^\OUT(w),d^\IN(y^i))$.
% if $w$ exists and at least $d^\OUT(y^o)$ otherwise.

If $d(a,H)\leq h$,
%$BFS^\OUT(a,h) \cap H \neq \emptyset$
then $d^\OUT(y^o)$ is at least $2h+z$ and the estimate is of the desired quality.
So assume that $d(a,H)> h$, and hence $d(w,H)\geq d(a,H)\geq h+1$. Thus $h'\geq h+1$, as we also have $\tilde{h}\geq h$ by assumption.
%, we must have
%%let $w$ be an arbitrary node that satisfies
%$BFS^\OUT(w,h) \cap H = \emptyset$.
%; the vertex furthest away from $H$ is such a node.
Assume also that $BFS^\OUT(w)$ is of depth at most $2h+z-1$ as if it is of depth at least $2h+z$ then the estimate is of the desired quality. Then, there is a vertex $w'\in BFS^\OUT(w,h')$ on the shortest path from $w$ to $b$ with $d(w,w')=h+1$ and hence $d(w',b)\leq h+z-2$.
%$w'\in BFS^\IN(b,h+z-2)\cap BFS^\OUT(w,h+1)$.
%As $d(w',b) \leq h+z-2$ w
As $d(a,b)=3h+z$, we must also have $d(a,w') \geq 2h+2$ and as $d^\IN(y^i) \geq d(a,w')$, the estimate is of the desired quality.
\end{proof}

Next, we analyze the running time of the algorithm.

%\begin{lemma}\label{L-run-time-sparse}
%Let $G=(V,E)$ be a graph of diameter $D=3h+z$, where $h\geq 0$ and $z\in \{0,1,2\}$. If $h$ is known,
%Approx-Diam-Sparse($G,h$) runs in $O(m^2/\Delta + \Delta^{h+1}m)$ time.
%\end{lemma}
%\begin{proof}
%The algorithm computes a BFS tree for every vertex of $H$.
%%As in $H$ we have only vertices with degree at least $\Delta$ it follows that $|H|=m/\Delta$. Thus, the cost of this part is $\Ot(m^2/\Delta)$.
%$|H|=O(m/\Delta)$ since there are at most that many vertices of degree at least $\Delta$. Hence the BFS computation from $H$ takes $O(m^2/\Delta)$ time.
%
%Computing the distances of the nodes in $V\setminus H$ to $H$ takes only $O(m)$ time. If $h$ is known, finding a node $w$ with  $BFS^\OUT(w,h) \cap H = \emptyset$, takes $O(n)$ time.
%Next, if such a $w$ exists, then the algorithm computes $BFS^\OUT(w)$ and $BFS^\IN(v)$ for every $v \in BFS^\OUT(w,h+1)$.
%As $BFS^\OUT(w,h) \cap H = \emptyset$, every vertex of $BFS^\OUT(w,h)$ is of degree at most $\Delta$, thus,
%$|BFS^\OUT(w,h+1)| \leq \Delta^{h+1}$ and the cost of this part of the algorithm is $O(\Delta^{h+1}m)$.
%\end{proof}
%

\begin{lemma}\label{L-run-time-sparse}
Let $G=(V,E)$ be a graph of diameter $D=3h+z$, where $h\geq 0$ and $z\in \{0,1,2\}$. If $\tilde{h}\geq h$,
Approx-Diam-Sparse($G,\tilde{h}$) runs in $O(m^2/\Delta + \Delta^{\tilde{h}+1}m)$ time.
\end{lemma}
\begin{proof}
The algorithm computes a BFS tree for every vertex of $H$.
%As in $H$ we have only vertices with degree at least $\Delta$ it follows that $|H|=m/\Delta$. Thus, the cost of this part is $\Ot(m^2/\Delta)$.
$|H|=O(m/\Delta)$ since there are at most that many vertices of outdegree at least $\Delta$. Hence the BFS computation from $H$ takes $O(m^2/\Delta)$ time.

Computing the distances of the nodes in $V\setminus H$ to $H$ takes only $O(m)$ time. Picking the node $w$ at largest distance to $H$ takes $O(n)$ time.
%If $h$ is known, finding a node $w$ with  $BFS^\OUT(w,h) \cap H = \emptyset$, takes $O(n)$ time.
%Next, if such a $w$ exists, then the algorithm computes $BFS^\OUT(w)$ and $BFS^\IN(v)$ for every $v \in BFS^\OUT(w,h+1)$.
%As $BFS^\OUT(w,h) \cap H = \emptyset$, every vertex of $BFS^\OUT(w,h)$ is of degree at most $\Delta$, thus,
%$|BFS^\OUT(w,h+1)| \leq \Delta^{h+1}$ and the cost of this part of the algorithm is $O(\Delta^{h+1}m)$.
The algorithm computes $BFS^\OUT(w)$ in $O(m)$ time. It then computes $BFS^\IN(v)$ for every $v \in BFS^\OUT(w,h')$ where $h'\leq \tilde{h}+1$.
%\leq 2h+2$.
Since we also have that $h'\leq d(w,H)$, every $v \in BFS^\OUT(w,h'-1)$ has outdegree at most $\Delta$. Thus, $|BFS^\OUT(w,h')| \leq \Delta^{h'}\leq \Delta^{\tilde{h}+1}$.
The running time of computing $BFS^\IN(v)$ for every $v \in BFS^\OUT(w,h')$ is hence at most $O(m\Delta^{\tilde{h}+1})$.
\end{proof}

%Lemmas~\ref{L-estimate-sparse} and~\ref{L-run-time-sparse} show that if $h$ is known, then an estimate of at least $2h+z$ can be returned in $\Ot(m^2/\Delta + \Delta^{h+1}m)$ time. Typically, we do not know $h$, and getting a good estimate on $h$ is in a sense equivalent to the problem we want to solve-- estimating $3h+z$. We can still get a good running time however by using a worse estimate on $h$ obtained in linear time.
We now prove Theorem~\ref{thm:sparsebetter} from the introduction.

\begin{reminder}{Theorem~\ref{thm:sparsebetter}}
There is an $\tilde{O}(m^{2-1/(2h+3)})$ time deterministic algorithm that computes an estimate $\hat{D}$ with $\lceil 2D/3\rceil \leq \hat{D}\leq D$ for all $m$-edge unweighted graphs of diameter $D=3h+z$  with $h\geq 0$ and $z\in \{0,1,2\}$. In particular, $\hat{D}\geq 2h+z$.
\end{reminder}

\begin{proof}
In $O(m)$ time we can get a $2$-approximation to the diameter, i.e. an estimate $E$ with $D/2\leq E\leq D$. Since $D=3h+z$, we have that
$(E-2)/3\leq h\leq 2E/3$. Setting $\tilde{h}=2E/3$ guarantees that $h\leq \tilde{h}\leq 2h+4/3<2h+2$, and hence $h\leq \tilde{h}\leq 2h+1$.

The quality of the estimate follows from Lemma~\ref{L-estimate-sparse} and by Lemma~\ref{L-run-time-sparse}, the runtime is $O(m^2/\Delta+m\Delta^{2h+2})$.
%We analyze the running time Algorithm~\ref{A-approx-m}($G,\tilde{h})$. The approximation guarantee from Lemma~\ref{L-estimate-sparse} still holds since we only compute more with an estimate higher than $h$.
%
%As in Lemma~\ref{L-run-time-sparse}, the BFS computation from $H$ takes $O(m^2/\Delta)$ time, and
%computing the distances of the nodes in $V\setminus H$ to $H$ takes only $O(m)$ time.
%
%Picking the node $w$ at largest distance to $H$ takes $O(n)$ time.
%The algorithm computes $BFS^\OUT(w)$ in $O(m)$ time. It then computes $BFS^\IN(v)$ for every $v \in BFS^\OUT(w,h')$ where $h'\leq \tilde{h}+1\leq 2h+2$.
%Since we also have that $h'\leq d(w,H)$, every $v \in BFS^\OUT(w,h'-1)$ has outdegree at most $\Delta$. Thus, $|BFS^\OUT(w,h')| \leq \Delta^{h'}\leq \Delta^{2h+2}$. The running time of computing $BFS^\IN(v)$ for every $v \in BFS^\OUT(w,h')$ is hence at most $O(m\Delta^{2h+2})$.
%
%The overall running time becomes $O(m^2/\Delta + m\Delta^{2h+2})$.
Picking $\Delta = m^{1/(2h+3)}$ minimizes the running time at $O(m^{2-1/(2h+3)})$.
\end{proof}

\paragraph{Acknowledgements}

The first author wants to thank Edith Cohen, Haim Kaplan and Yahav Nussbaum for fruitful discussions on the problem.
The second author wants to thank Bob Tarjan for asking whether there is an almost linear time approximation scheme for the diameter. 

%\newpage
\bibliographystyle{plain}
\bibliography{diam}
\newpage

\section{Appendix} 

\begin{algorithm}[ht]\label{A-approx}
\caption{Approx-Diam($G$)}
$X_1 \gets$ Aingworth($G$)\;
$X_2 \gets$ Aingworth($G^R$)\;

$\hat{D} \gets \max (X_1,X_2)$\;

\ForEach{$v\in V$}{\ForEach{$u\in V\setminus \{ v \}$}
{

\If{$BFS^\OUT(u,d^\OUT_s(u)-1)\cap BFS^\IN(v,d^\IN_s(v)-1)=\emptyset \wedge \nexists (u',v')\in E$ s.t. $u'\in BFS^\OUT(u,d^\OUT_s(u)-1) \wedge v'\in BFS^\IN(v,d^\IN_s(v)-1)$}{$\hat{D} \gets \max (\hat{D},d^\OUT_s(u)+d^\IN_s(v))$}}}
\Return $\hat{D}$\;
\end{algorithm}

\begin{algorithm}[h]\label{A-approx-m}
\caption{Approx-Diam-Sparse($G,\tilde{h}$)}

$H \gets \{ v \mid deg(v) \geq \Delta \}$\;
\lForEach{$ v\in H$}{Compute $BFS^\OUT(v)$\;}
$y^o \gets \arg\max_{x\in H} d^\OUT(x)$\;
$\hat{D}\gets d^\OUT(y^o)$\;

%
%\lForEach{$ v\in V$}{Compute $N_s^\OUT(v)$\;}
%Compute a set $S^\OUT$ that hits every out neighborhood of size $s$\;
%$x^o \gets \arg\max_{x\in V} d_s^{\OUT}(x)$\;
%\lForEach{$ v\in S^\OUT$}{Compute $BFS^\OUT(v)$\;}
%
%$y^o \gets \arg\max_{x\in S^\OUT} d^\OUT(x)$\;
%
%\lIf{$d_s^\OUT(x^o)\leq h$}{\Return $d^\OUT(y^o)$\;}

Compute $d(v,H)$ for all $v\in V$ with a single BFS\;
$w \gets $ vertex of largest $d(w,H)$\;
%Compute $F = \{ w \mid d(w,p_\Delta(w))>h\}$ with a single BFS\;
%\lIf{$F =\emptyset$}{\Return $d^\OUT(y^o)$\;}
%$w \gets $ an arbitrary vertex of $F$\;
Compute $BFS^\OUT(w)$\;
$\hat{D}\gets \max\{\hat{D}, d^\OUT(y^o)\}$\;

$h'\gets \min \{\tilde{h}+1, d(w,H)\}$\;
\lForEach{$ v\in BFS^\OUT(w,h')$}{Compute $BFS^\IN(v)$\;}
$y^i \gets \arg\max_{x\in BFS^\OUT(w,h')} d^\IN(x)$\;
$\hat{D}\gets \max\{\hat{D}, d^\IN(y^i)\}$\;

\Return $\hat{D}$\;
\end{algorithm}

\end{document}